\documentclass[12pt]{article} 

\usepackage{times}
\usepackage{amssymb}
\usepackage{amsmath}
\usepackage{graphicx}
\usepackage{epsfig}
\usepackage{subfigure}
\usepackage{color}
\usepackage{authblk}

\def\shadowbox{\hbox{\rule[-0.0ex]{0.1ex}{1.2ex}%
\hspace{-0.1ex}\rule[-0.0ex]{1.2ex}{0.1ex}%
\hspace{0.0ex}\rule[-0.0ex]{0.1ex}{1.2ex}\hspace{-1.3ex}%
\rule[1.15ex]{1.25ex}{0.1ex}\hspace{-0.0ex}\rule[-0.25ex]{0.3ex}{1.1ex}%
\hspace{-1.2ex}\rule[-0.25ex]{1.1ex}{0.25ex}}}
\def\qed{\ifmmode \hbox{\hfill\shadowbox}
     \else \hphantom{x}\hfill\shadowbox \fi}

\newtheorem{theorem}{Theorem}[section]
\newtheorem{lemma}[theorem]{Lemma}
\newtheorem{definition}[theorem]{Definition}
\newtheorem{proposition}[theorem]{Proposition}
\newtheorem{corollary}[theorem]{Corollary}
\newtheorem{remark}[theorem]{Remark}

\def\Z {\mathbb{Z}}
\def\N {\mathbb{N}}
\def\F {\mathcal{F}}
\def\E {\mathcal{E}}
\def\T {\mathbb{T}}

\def\supp {\textnormal{supp}}

\def\N {\mathbb{N}}
\def\C {\mathbb{C}}

\def\Izo {\int_0^1}
\def\skI {\sum_{k\in I}}
\def\kt {\tilde{k}}
\def\M {\mathcal{M}}
\def\MrI {\mathcal{M}(w_r,I)}
\def\half {\frac{1}{2}}
\def\wo {w^{(1)}}

\def\sign {\textnormal{sign}}
\def\Ex {\textnormal{Ex}}
\def\calC {\mathcal{C}}
\def\span {\textnormal{span}}
\def\Interp {\textnormal{Interp}}
\def\Lt {L^2([0,1])}
\def\esssup {\textnormal{ess.sup}}
\def\pk {\phi_k}
\def\PE {\Psi_E}
\def\PAPR {\textnormal{PAPR}}
\allowdisplaybreaks[2]

\begin{document}
\title{On the Peak-to-Average Power Ratio Reduction Problem for Orthogonal Transmission Schemes}

\author[1]{Holger Boche\thanks{Theresienstr. 90, 80333 M\"unchen, Germany. Email: boche@tum.de}}
\author[2]{Brendan Farrell\thanks{1200 E. California Blvd, Pasadena, CA 91125, USA. Email:farrell@cms.caltech.edu\\
This work was completed while both authors were with the Technische Universit\"at Berlin and 
was supported by the German Science Foundation (DFG) under Project BO 1734/18-1.}}
\affil[1]{Technische Universit\"at M\"unchen}
\affil[2]{California Institute of Technology}
\date{}
\maketitle

\begin{abstract}
High peak values of transmission signals in wireless communication systems lead to wasteful energy consumption and out-of-band radiation. 
However, reducing peak values generally comes at the cost some other resource. 
We provide a theoretical contribution towards understanding the relationship between peak value reduction and the resulting cost in information rates. 
In particular, we address the relationship between peak values and the proportion of transmission signals allocated for 
information transmission when using a strategy known as tone reservation. 
We show that when using tone reservation in  both OFDM and DS-CDMA systems, if a Peak-to-Average Power Ratio criterion is always satisfied, 
then the proportion of transmission signals that may be allocated for information transmission must tend to zero. 
We investigate properties of these two systems for sets of both finite and infinite cardinalities. 
We present properties that OFDM and DS-CDMA share in common as well as ways in which they fundamentally differ. 
\end{abstract}

\begin{section}{Introduction}

Recent studies by two consulting firms have estimated  that 2\% of global $\textnormal{CO}_2$ emissions are attributable to the use of 
information and communication technology, 
a contribution comparable to aviation \cite{BLO08,Gar10,PST08}.  
While this impact is already significant, the amount of information communicated electronically is growing exponentially, 
and the emission percentage is expected to increase to 3\% by 2020. 
A large portion of the energy consumption causing emissions is due to wireless communications, 
and within wireless systems, a significant portion of the energy consumption occurs at the amplifiers. 
Communications companies must deal with a trade-off between expensive amplifiers that are efficient, a capital expenditure, or 
inexpensive amplifiers and high energy costs, an operating cost. 
The increase in the volume of wireless communication requires using systems that place more and more individual signals in a frequency band, 
and this, inherently, leads to larger signal amplitudes. 
Thus, from both an environmental as well as a financial perspective, the interplay between information capacity, signal peak values, amplifier 
performance and energy consumption is an essential area for research. 
The current understanding is that amplifiers are more efficient when transmitting signals that have smaller peak values. 
See \cite{RACKPPSS} for an overview. 

Within the broader question of power consumption, amplifier efficiency and capacity in mind, 
we focus in this paper on the relationship between signal peak values and the proportion of signal resources that 
can be allocated for information transmission. 
We address the balance between allocating resources towards reducing signal peak values and allocating them for transmitting information. 
To the best of the authors' knowledge, there has been very little theoretical work done here, and little is known about the fundamental 
relationships between these various aspects of wireless communication. 

We focus our attention on two of the most important contemporary communications systems, namely 
Orthogonal Frequency Division Multiplexing (OFDM) and Direct Sequence-Code Division Multiple Access (DS-CDMA). 
Both of these use a classical basis for transmission signals: 
OFDM uses the Fourier basis and (DS-CDMA) uses the Walsh basis. 
In either case, coefficients are chosen to represent a message, 
and the linear combination corresponding to these coefficients is transmitted. 
As part of transmission the signal passes through an amplifier. 
Every amplifier has a threshold beyond which it cannot linearly amplify the signal, 
but distorts or ``clips'' it. 
We say that a signal is clipped at magnitude $M$ if the signal is left undisturbed where its magnitude is less than 
$M$ and its magnitude is reduced to $M$ where it is greater than $M$ while leaving the phase unchanged. 
Both distortion and clipping take place only above some threshold, so that one may say that only the peaks are affected.

Transmission signals generally also satisfy a frequency band requirement, 
that is their Fourier transforms are supported in a specified region or the signal is band-limited. 
This gives that the signals are analytic, i.e. infinitely differentiable with convergent power series.  
If the distorted signal differs from the original only where the original magnitude exceeds a threshold, 
then the two signals agree where the original is below that threshold. 
If the distorted function is also band-limited, then it is analytic and the difference between original and distorted is 
zero on an open interval.
A basic theorem of complex analysis then implies that their difference is identically zero. 
Thus, the distorted function cannot be band-limited. 
That is, if clipping or distortion occurs, then the amplified signal is not band-limited, and out-of-band radiation occurs. 
This gives the motivation for the approach taken in this paper.  

If the transmitted signal is not band-limited it interferes with other frequency bands. 
If one is interested in capacity or error-rates without a requirement that out-of-band radiation does not occur, 
then a probabilistic approach to the peak behavior is appropriate. 
However, in many instances out-of-band radiation may be strictly prohibited.

With the transition from analog to digital television transmission, the improved efficiencies allowed new frequency bands to be 
redistributed, in particular for wireless communications. 
This is commonly called the ``Digital Dividend''. 
Very strict quality of service requirements have been imposed on operators for some of these bands, 
and a percentage based compliance is insufficient.  
An example is wireless microphones, where users are of course very sensitive to a disruption of service. 
In these cases statistical models are inadequate. 

The strategy we consider here is known in the OFDM setting as Tone Reservation. 
This method was introduced in~\cite{TC98,Tel99}; an overview is given in~\cite{HL05,Lit07} and a survey of 
recent advances is given in~\cite{WFBLN12}.
We will apply this strategy as well to the CDMA setting. 
Here one separates the available transmission signals into two subsets. 
Coefficients that carry the message are then applied to signals in one subset, and then coefficients are determined for the 
signals in the second subset, such that the peak of the entire composite signal is ideally below a threshold. 
There are other methods to reduce the peak value, such as selected mapping, clipping and filtering and selected mapping. 
However, tone reservation is canonical in that the coefficients to be transmitted are not altered in any way, 
the auxiliary coefficients may simply be ignored by the receiver, 
and there is no additional overhead to transmission. 
We note that the literature on the these topics is enormous with some papers cited several hundred times. 
Extensive references of the most important works are available in~\cite{Lit07,WFBLN12}.

Here we address how the ratio of information-bearing signals to compensation signals behaves with respect to a peak threshold 
as the total number of signals available increases. 
The two main results presented here, Theorem~\ref{discretenotsolvable} in the Fourier-OFDM case and Theorem~\ref{mainwalsh} in the Walsh-CDMA case, 
show that if a peak threshold must always be satisfied, then the proportion of signals that may be used to carry information converges to zero. 
While the OFDM and CDMA share this property in common, they behave quite differently in other significant ways, which we discuss in later sections. 
These main results are coupled with two other main points. 
This first is a relationship between what we will call \emph{solvability} and a norm equivalence, and the second is a fundamentally different 
behavior when these questions are addressed for sets of finite or infinite cardinality. 
The relationship between solvability and the norm equivalence is presented in Section~\ref{motif}. 
The Fourier-OFDM case is addressed in Section~\ref{Discrete}, 
and the density result for the Walsh-CDMA case is presented in Section~\ref{walsh}. 
Section~\ref{RW} gives further properties of the Walsh system and, in particular, emphasizes their localized nature. 

While this paper addresses Fourier-OFDM and Walsh-CDMA systems, we note that recent results for the peak-value behavior of 
$sinc$ or single-carrier systems were obtained in~\cite{BFLW12}. 
There it is shown that the expected peak of a random linear combination of shifted $sinc$ functions grows with the number of 
such functions. 
This underscores the prevalence of high peak amplitudes in communications systems. 

\end{section}

\begin{section}{Solvability and a Norm Equivalence for Orthonormal Systems}\label{motif}
\begin{subsection}{Introductory Facts}
We first formalize our problem and then introduce an important aspect of our approach. 
We begin with the following definition, where, without loss of generality, we take $[0,1]$ as the symbol interval.  
\begin{definition}
Given a set of orthonormal functions $\{\phi_n\}_{n=1}^N \subset \Lt$, we define the \emph{Peak-to-Average Power Ratio (PAPR)} of a 
set of coefficients $a\in \C^N$ by
\begin{equation}
\PAPR(\{\phi_n\}_{n=1}^N,a)=\esssup_{t\in [0,1]}\frac{|\sum_{n=1}^N a_n\phi_n(t)|}{\|a\|_{l^2_N}}.
\end{equation}
\end{definition}
The following simple proposition shows that PAPR values of order $\sqrt{N}$ can occur for any orthonormal system. 
We include the proof just for the sake of completeness. 
\begin{proposition}\emph{(Theorem 6 in \cite{BP07})}\label{sqN}
Let $\{\phi_k\}_{k=1}^N$ be $N$ orthonormal functions in $L^2([0,1])$. 
Then there exists a sequence $a\in l^2_N$ with norm $\|a\|_{l^2_N}=1$ such that 
\begin{equation*}
\esssup_{t\in [0,1]}|\sum_{n=1}^N a_n \phi_n(t)|\geq \sqrt{N}. 
\end{equation*}
\end{proposition}

\begin{proof}
First we observe 
\begin{equation*}
N=\frac{1}{2\pi}\int_{-\pi}^{\pi} \sum_{n=1}^N |\phi_n(t)|^2dt\leq \esssup_{t\in [0,1]}\sum_{n=1}^N |\phi_n(t)|^2.
\end{equation*}
So, for any $\epsilon>0$, there exists $t_0\in[0,1]$ such that all $\{\phi_k\}_{k=1}^N$ are defined at $t_0$ and 
\begin{equation}
N-\epsilon\leq \sum_{n=1}^N |\phi_n(t_0)|^2.\label{4444}
\end{equation}
Now set 
\begin{equation*}
a_n=\frac{\overline{\phi_n(t_0)}}{\sqrt{  \sum_{n=1}^N |\phi_n(t_0)|^2}}.
\end{equation*}
Using inequality~(\ref{4444}) we have  
\begin{eqnarray*}
\sup_{t\in[0,1]} \sum_{n=1}^Na_n \phi_n(t)&\geq &\sum_{n=1}^N  a_n \phi_n(t_0)
= \frac{ \sum_{n=1}^N\overline{\phi_n(t_0)}\phi_n(t_0) }{ \sqrt{  \sum_{n=1}^N |\phi_n(t_0)|^2}  }\\
&=&\left( \sum_{n=1}^N |\phi_n(t_0)|^2\right)^{1/2}\geq  \sqrt{N-\epsilon}.
\end{eqnarray*}
Since $\epsilon$ is arbitrary, we have proved the proposition. 
\end{proof}
Thus, for any orthonormal basis $\{\phi_n\}^\infty_{n=1}$, we have 
\begin{equation}
  \sup_{\|a\|_{l^2}=1}\PAPR(\{\phi_n\}^N_{n=1},a)\geq \sqrt{N}. 
\end{equation}
In fact, $\sqrt{N}$ is also a bound on the PAPR for both the OFDM and the DS-CDMA systems. 
  Since the transmission signals in each of these cases are uniformly bounded by one, this follows from 
applying the Cauchy-Schwarz inequality pointwise to the linear combination. 
Therefore, OFDM does not offer any advantages as far as worst-case performance for PAPR. 
Proposition~\ref{sqN} shows that the upper bound on PAPR for these two systems is also a lower bound on PAPR for all orthonormal systems.  
\end{subsection}

\begin{subsection}{Solvability and a Norm-Equivalence}

We recall from the Introduction that the strategy addressed in this paper is to reserve one subset of orthonormal functions for carrying the information-bearing 
coefficients and to determine coefficients for the remaining orthonormal functions, so that the combined sum of functions has 
a small peak value. 
We formalize this in the following definition. 
\begin{definition}\label{solvable}
The PAPR reduction problem is \emph{solvable} for the orthonormal system $\{\phi_n\}_{n=1}^\infty$ and the subset $K\subset \N$ with 
constant $C_{\Ex}$ if for every $a\in l^2(K)$ there exists $b\in l^2(K^c)$, satisfying $\|b\|_{l^2(K^c)}\leq C_{\Ex} \|a\|_{l^2(K)}$ such that 
\begin{equation}
\esssup_{t\in [0,1]} \left|\sum_{n\in K}a_n\phi_n(t)+\sum_{n\in K^c}b_n\phi_n(t)\right|\leq C_{\Ex} \|a\|_{l^2(K)}.
\end{equation}
\end{definition}

We may view the map from the coefficient vector $a$ to 
a function with a small peak as an \emph{extension operator}. 
This operator is a map from $l^2(K)$ to $\Lt$ given by 
\begin{equation}
E_Ka=\sum_{n\in K}a_n\phi_n(t)+\sum_{n\in K^c}b_n\phi_n(t).
\end{equation}
Note that this map is not necessarily unique and is generally not linear; 
we will also \emph{not} discuss the construction of such a map. 
Nonetheless, since the map gives a correspondence between $l^2(K)$ and $\Lt$ we refer to it as the extension operator. 
Thus, we equivalently say that the PAPR reduction problem is solvable for $\{\phi_n\}_{n=1}^\infty$ and $K$ with extension norm  
$C_{\Ex}$ if there exists an extension operator $E_K$ such that 
\begin{equation}
\| E_K\|_{l^2(K)\rightarrow \Lt}\leq C_{\Ex}.
\end{equation}
Note that we are only interested in the existence of \emph{an} extension operator, and that uniqueness is not part of the discussion. 
Clearly the operator is generally not linear. 

The main results of this paper concern the proportion of signals that may be used for information transmission under a peak value constraint. 
Our approach, however, builds on a further point, 
namely a relationship between PAPR reduction and an $L^1-L^2$ norm equivalence. 
Given an orthonormal system $\{\phi_n\}_{n=1}^\infty$ for $\Lt$ and a subset $K\subset \N$, we define 
\begin{equation}
X:=\left\{ f:f\in L^1([0,1]),\;f=\sum_{n\in K}a_n\phi_n\right\}.
\end{equation}
The torus is defined by 
\begin{equation*}
\T=\{z\in \C:\;|z|=1\}. 
\end{equation*}

\begin{theorem}\emph{(\cite{BF10})}\label{equivalence}
Let $\{\pk\}_{k\in\N}$  be an orthonormal basis for $\Lt$, let $K$ be a subset of $\N$, and let $X$ be as just defined.  
The PAPR problem is solvable for the pair $K$ and $\{\pk\}_{k\in\N}$  with extension norm $C_{\Ex}$
if and only if 
\begin{equation}
\|f\|_{L^2(\T)}\leq C_{\Ex}\|f\|_{L^1(\T)}\label{ne}
\end{equation}
for all $f\in X$. 
\end{theorem}
While Theorem~\ref{equivalence} is proved in~\cite{BF10}, we include the proof here  so that the role of the Hahn-Banach 
Theorem is apparent. In particular, the proof relies on the existence of a function, denoted $r$ below, 
for which, in general, there does not exist a method to construct. 

\begin{proof}
i.) Assume that the PAPR problem is solvable. 
Then for all $s(t)=\sum_{k\in K} a_k\pk(t)$, $\|a\|_{l^2(\Z)}\leq 1$, 
\begin{equation}
\|E_K a\|_{L^\infty(\T)} \leq C_{\Ex}\|s\|_{L^2(\T)}\leq C_{\Ex}.
\end{equation}
Since $L^\infty(\T)\subset L^2(\T)$, 
\begin{equation}
E_Ks=\sum_{k\in K}a_k\pk+\sum_{k\in \N\backslash K}b_k\pk.
\end{equation}
Let $f\in X$, $f(t)=\sum_{k\in K} c_k\pk(t)$, be arbitrary. 
Then 
\begin{eqnarray*}
\left|\sum_{k\in K}a_k\overline{c}_k\right|&=&\left|\sum_{k\in K}a_k\overline{c}_k + \sum_{k\in \Z\backslash K}b_k\overline{c}_k \right|\\
&=&\left|\frac{1}{2\pi}\int_\T f(t)\overline{E_K s(t)}dt\right|\\
&\leq& \|f\|_{L^1(\T)}\|E_Ks\|_{L^\infty(\T)}\\
&\leq& C_{\Ex}\|f\|_{L^1(\T)}.
\end{eqnarray*}
Set 
\begin{equation*}
a_k=\Big\{
\begin{array}{cl}
\frac{c_k}{\|c\|}_{l^2}&c_k\neq 0\\
0&c_k=0
\end{array}.
\end{equation*}
Then $\|f\|_{L^2(\T)}=\|c\|_{l^2}=|\sum_{k\in K}a_k\overline{c}_k| \leq C_{\Ex}\|f\|_{L^1(\T)}$.

ii.) Assume $\|f\|_{L^2(\T)}\leq C_{\Ex}\|f\|_{L^1(\T)}$ for all $f\in X$. 
Let $a\in l^2(\Z)$ be a sequence, supported in $K$, with only finitely many nonzero terms satisfying $\|a\|_{l^2(\Z)}\leq 1$. 
Set $s(t)=\sum_{k\in K}a_k \pk(t)$. 
For $f\in X$, $f(t)=\sum_{k\in K}c_k \pk(t)$,    define the functional $\Psi_a$  by 
\begin{equation*}
\Psi_a f=\sum_{k\in K}a_k\overline{c}_k.
\end{equation*}
Since 
\begin{equation*}
|\Psi_a f|\leq \|a\|_{l^2(\Z)}\|c\|_{l^2(\Z)}\leq \|f\|_{L^2(\T)}\leq C_{\Ex}\|f\|_{L^1(\T)}, 
\end{equation*}
$\Psi_a$ is continuous on $X$. 
Since $X$ is a closed subspace of $L^1(\T)$, 
by the Hahn-Banach Theorem,
the functional $\Psi_a$ has the extension $\PE$ to all of $L^1(\T)$, 
where $\|\Psi_a\|=\|\PE\|$.  
The dual of $L^1(\T)$ is $L^\infty(\T)$. 
Thus, for some $r\in L^\infty(\T)$,  
\begin{equation*}
\PE f=\langle f, r\rangle, 
\end{equation*}
for all $f\in L^1(\T)$, 
so that $\|\PE\|=\|r\|_{L^\infty(\T)}$. 
Since $L^\infty(\T)\subset L^2(\T)$, $r$ possesses the unique expansion 
\begin{equation*}
r(t)=\sum_{k\in\N}d_k\pk(t)
\end{equation*}
for some $d\in l^2(\Z)$. 
The sequences $d$ and $a$ agree on $K$, and we define $E_Ks:=r$.
\end{proof}
We will also address the case when we have a finite set of basis functions intended for information and a finite set reserved for 
peak reduction. 
We then have a finite set  $\{\pk\}_{k=1}^N$, which of course is then not an orthonormal basis for $\Lt$. 
Consequently, we in general cannot represent the function $r$ in the proof above in terms of $\{\pk\}_{k=1}^N$. 
Nonetheless, we have one direction of Theorem~\ref{equivalence}, which we state as a corollary.  
\begin{corollary}\label{corequiv}
Let $\{\pk\}_{k\in\N}$  be a set of orthonormal functions in $\Lt$, let $K$ be a subset of $\N$, and let $X$ be as previously defined.  
If the PAPR problem is solvable for the pair $K$ and $\{\pk\}_{k\in\N}$  with extension norm $C_{\Ex}$
then  
\begin{equation}
\|f\|_{L^2(\T)}\leq C_{\Ex}\|f\|_{L^1(\T)}\label{repeat}
\end{equation}
for all $f\in X$. 
\end{corollary}
In the finite setting, if one can show that a constant $\C_{\Ex}$ does not exist such that 
the norm property in line (\ref{repeat}) holds for all finite cardinalities, then it follows that solvability cannot hold. 
Thus, to prove that solvability does not hold, we do not use the two-way statement of Theorem~\ref{equivalence}, but just this one-way statement.  
This will be our approach later.

\end{subsection}
\end{section}

\begin{section}{The Discrete Fourier Case}\label{Discrete}
\begin{subsection}{Density of Information Sets for OFDM }
The discrete Fourier case is interesting for several reasons. 
First, the discrete case implies the continuous case, and therefore, delivers the result on the density of tone reservation sets for OFDM. 
Additionally, the problem considered is applicable to a large number of areas and is valuable in its own right. 
The discrete case is important for the PAPR problem because much of the work done with signals is, of course, done 
with discretized versions of the signals. 
For example, oversampling and zero-padding are used in the papers~\cite{Tel01,IS09}. 
In some settings it is possible, using sampling results, to relate discrete properties to analog properties, 
and therefore it is valuable to understand the behavior in the discrete setting. 
In \cite{WB03}, for example, bounds on the PAPR of an OFDM signal are given in terms of samples of the signal and the over-sampling rate. 
\begin{definition}
The $N\times N$ inverse discrete Fourier transform (DFT) matrix is given by 
\begin{equation*}
F_{jk}=\frac{1}{\sqrt{N}}e^{-2\pi i (j-1)(k-1)/N}.
\end{equation*}
This matrix is denoted $F$, and for $x\in l^2_N$, $Fx$ denotes this matrix applied to $x$. 
  $F^*$ denotes the Hermitian transpose or adjoint of $F$. 
\end{definition}

\begin{definition}
$l^p_N$ denotes $\C^N$ viewed as a linear space with norm $\|x\|_{l^p_N}=(\sum_{i=1}^N |x_i|^p)^{1/p}$. 
The unit ball in $l^p_N$ is denoted $B^p_N$, i.e. 
\begin{equation*}
B^p_N=\{ x\in l^p_N:\; \|x\|_{l^p_N}\leq 1\}.
\end{equation*}
\end{definition}

Let $\{N_k\}_{k=1}^\infty$ be a subsequence of $\N$, and for each $N_k$ let $I_{N_k}$ be a subset of $\{1,...,N_k\}$. 
$I_{N_k}^c $ denotes $\{1,...,N_k\}\backslash I_{N_k}$. 
In analogy to Definition~\ref{solvable}, we say the discrete PAPR problem is \emph{solvable}  for the sequences $\{N_k\}_{k=1}^\infty$ and  $\{I_{N_k}\}_{k=1}^\infty$ 
if there exists a constant $C_{\Ex}$, such that for each $k$, for all $x\in l^2_{N_k}$ with $\supp(x)\subset I_{N_k}$ there exists a 
compensation vector $r\in l^2_{N_k}$ supported in $I_{N_k}^c$ such that 
\begin{equation*}
\|F(x+r)\|_{l^\infty_{N_k}} \leq \frac{C_{\Ex}}{\sqrt{N_k}}\|x\|_{l^2_{N_k}}. 
\end{equation*}

\begin{theorem}\label{equivdiscrete}
Let $\{N_k\}_{k=1}^\infty$ be a subsequence of $\N$, 
and let $I_{N_k}$ be a subset of $\{1,...,N_k\}$. 
Let $Y_k=\{y\in l^2_{N_k}:\; \supp(F^*y)\subset I_{N_k}\}$. 
The discrete PAPR problem is solvable for the sequence of sets $\{I_{N_k}\}_{k=1}^\infty$ 
with constant $C_{\Ex}$ if and only if 
\begin{equation*}
\| y\|_{l^2_{N_k}}\leq \frac{C_{\Ex}}{\sqrt{N_k}} \|y\|_{l^1_{N_k}}
\end{equation*}
for all $y\in Y_k$. 
\end{theorem}
In general one has $\|y\|_{l^2_{N_k}}\leq \|y\|_{l^1_{N_k}}$ for any vector $y$. 
Here, though, as $k$ increases, we eventually have $C_{\Ex}/\sqrt{N_k}<1$, 
and so the important point is that $C_{Ex}$ remains fixed. 

\begin{proof}

i.) Assume that PAPR is solvable with constant $C_{\Ex}$. Let $N$ be an element of $\{N_k\}_{k=1}^\infty$. 
For $I_N\subset \{1,...,N\}$, let $I_N^c=\{1,...,N\}\backslash I_N$. 
Then for any $x\in \C^N$ with $\supp(x)\subset I_N$, we can find an 
extension $r\in \C^N$ with $\supp(r)\subset I_N^c$, such that 
\begin{equation*}
F(x+r)\in \frac{C_{\Ex}}{\sqrt{N}} B^{\infty}_N. 
\end{equation*}
We denote by $l^p_N(I_N)$ elements of $l^p_N$ with support contained in $I_N$. 
Denote by $E_{I_N}$ the operator that maps $x$ to the compensated vector $x+r$. 
Then 
\begin{equation}
\|FE_{I_N}x\|_{l^\infty_N}\leq \frac{C_{\Ex}}{\sqrt{N}} \|x\|_{l^2_N}, \label{discretevector}
\end{equation}
and so $\|F E_{I_N}\|_{l^2_N(I_N)\rightarrow l^\infty_N}\leq \frac{C_{\Ex}}{\sqrt{N}}$. 
As in the analog case, we take a vector $b$ with $\supp(b)\subset I_N$, 
and observe 
\begin{eqnarray*}
|\langle b,E_{I_N}x\rangle |&=& |\langle F b,FE_{I_N}x\rangle |\\
&\leq & \|Fb\|_{l^1_N} \|FE_{I_N}x\|_{l^\infty_N}\\
&\leq & \|Fb\|_{l^1_N} \frac{C_{\Ex}}{\sqrt{N}} \|x\|_{l^2_N}.
\end{eqnarray*}
By setting 
\begin{equation*}
x_k=\bigg\{
\begin{array}{cl}
\frac{b_k}{\|b\|}_{l^2_N}&b_k\neq 0\\
0&b_k=0
\end{array},
\end{equation*}
we obtain
\begin{eqnarray*}
\|b\|_{l^2_N}&=& |\langle b,b\rangle|^2\\
&=& |\langle b,E_{I_N}x\rangle|\\
&\leq & \frac{C_{\Ex}}{\sqrt{N}} \|Fb\|_{l^1_N}. 
\end{eqnarray*}

ii.) 
Let $N$ be an element of $\{N_k\}_{k=1}^\infty$. 
We take an element $c\in l^2_N$ with  $\supp(c)\subset I_N$. 
Let $\Psi_c$ be the functional acting on $Y$ by 
\begin{equation}
\Psi_c y= \langle c,F^* y\rangle.\label{123}
\end{equation}
We then have $|\Psi_cy|\leq \|c\|_{l^2_N} \|y\|_{l^2_N}\leq \frac{C_{\Ex}}{\sqrt{N}} \|y\|_{l^1_N} \|c\|_{l^2_N}$, 
so that 
\begin{equation}
\|\Psi_c\|\leq \frac{C_{\Ex}}{\sqrt{N}}\|c\|_{l^2_N}.\label{1234}
\end{equation} 
Since $Y$ is a closed subspace of $l^1_N$, by the Hahn-Banach Theorem there exists 
an extension $\Psi_E$ of $\Psi_c$ to all of $l^1_N$ such that $\|\Psi_c\|=\|\Psi_E\|$. 
$\Psi_E$ can be represented by a vector $r$ so that 
\begin{equation*}
\Psi_E y= \langle r,y\rangle
\end{equation*}
for all $y\in l^1_N$. 
Let $\overline{c}=Fr$. 
If $y\in Y$ and $y=Fx$, 
then $\langle r,y\rangle= \langle F^* \overline{c},F^*x\rangle = \langle \overline{c},x\rangle$. 
Comparing this with equation~(\ref{123}), 
we see that $c$ and $\overline{c}$  must agree on $I_N$. 
That is, $\overline{c}$ is an extension of $c$. 
Lastly,  using equation~(\ref{1234}),
\begin{eqnarray*}
\| \Psi_E\|&=& \|r\|_\infty\\
&=& \|F\overline{c} \|_\infty\\
&\leq & \frac{C_{\Ex}}{\sqrt{N}}\|c\|_{l^2_N}. 
\end{eqnarray*}

\end{proof}

For a set $A$, $|A|$ denotes its cardinality.

\begin{theorem}\label{discretenotsolvable}
Let $\{N_k\}_{k=1}^\infty$ be a subsequence of $\N$ and let $I_{N_k}$ be the corresponding sets as defined earlier. 
If 
\begin{equation*}
\limsup_{n\rightarrow \infty} \frac{|I_{N_k}|}{N_k}>0,
\end{equation*}
then the discrete PAPR problem is not solvable. 
\end{theorem}

The proof will use arithmetic progressions and Szemer\'edi's Theorem, Theorem~\ref{SzTao}. 
\begin{definition}
An \emph{arithmetic progression of length} $m$ is a subset of $\Z$ that has the form 
$\{a, a+d,a+2d,....,a+(m-1)d\}$ for some integer $a$ and some positive integer $d$. 
\end{definition}

\begin{theorem}(Theorem 1.2 in \cite{Tao06})\label{SzTao}
For any integer $m\geq 1$ and any $0<\delta\leq 1$, there exists an integer $N_{Sz}(m,\delta)\geq 1$ such that for 
every $N\geq N_{Sz}(m,\delta)$, 
every set $A\subset \{1,...,N\}$ of cardinality $|A|\geq \delta N$ contains at least one arithmetic progression of length $m$. 
\end{theorem}

\begin{proof}{\bf of Theorem~\ref{discretenotsolvable}} 
By Theorem~\ref{SzTao}, there exists an integer $N$ in the subsequence $\{N_k\}_{k=1}^\infty$ such that $I_{N}$ contains 
an arithmetic progression of length $m$. 
Assume again that this progression is $\{a+dl\}_{l=0}^{m-1}$. 
Let $D$ denote the vector of length $N$ with the value $e^{\frac{2\pi i(a+dl)t}{N}}/\sqrt{m}$ 
at the entries of the arithmetic progression, where $t$ will be 
addressed shortly.  
Then 
\begin{eqnarray}
\|FD\|_{l^1_{N}}&=& \sum_{j=1}^N|(FD)_j|\nonumber\\
&=& \sum_{j=1}^N\left| \sum_{l=1}^NF_{jl}D_l\right|\nonumber\\
&=&  \sum_{j=1}^N\left|\frac{1}{\sqrt{N}}\frac{1}{\sqrt{m}}\sum_{l=0}^{m-1}e^{-\frac{2\pi i(a+dl)t}{N} }e^{\frac{2\pi i dlj}{N}}\right|\nonumber\\
&=&\frac{1}{\sqrt{m}}\frac{1}{\sqrt{N}} \sum_{j=1}^N\left|\sum_{l=0}^{m-1}e^{\frac{2\pi i dl(t-j)}{N}}\right|.\label{apr}
\end{eqnarray}
This calculation holds for any $t$, so we may take the $t$ that minimizes the 
absolute value: 
\begin{eqnarray*}
\min_{t\in[0,1]} \sum_{j=1}^N\left|\sum_{l=0}^{m-1}e^{\frac{2\pi i dl(t-j)}{N}}\right|
&=& \min_{t\in[0,1]} \sum_{j=1}^N \left|\frac{\sin\frac{\pi dm(t-j)}{N}}{\sin\frac{\pi d(t-j)}{N}}\right|\\
&\leq &  \int_{0}^1 \sum_{j=1}^N \left|\frac{\sin\frac{\pi dm(t-j)}{N}}{\sin\frac{\pi d(t-j)}{N}}\right|dt \\
&= & N \int_{0}^1 \left|\frac{\sin\frac{\pi dmt}{N}}{\sin\frac{\pi dt}{N}}\right|dt\\
&\leq & N\log m,
\end{eqnarray*}
where the last step is the bound on the Dirichlet kernel. 
Now, returning to line~(\ref{apr}), 
and defining $D$ using the $t$ that results in the minimum in the calculation above, 
we have $\|FD\|_{l^1_{N}} \leq \frac{\log m}{\sqrt{m}}\sqrt{N}$. 
If the discrete PAPR problem is solvable, then by Theorem~\ref{equivdiscrete}, we have a norm equivalence with a factor $C_{\Ex}/\sqrt{N}$. 
However, we have just shown that $C_{\Ex}$ must be arbitrarily small. 
This contradiction proves Theorem~\ref{discretenotsolvable}.  
\end{proof}

Our next result shows that one can have a norm equivalence on a subspaces given by subsets of the 
columns of the DFT matrix when the density converges to zero fast enough. 

\begin{corollary}
Let $\{N_k\}_{k=1}^\infty$ be a subsequence of $\N$ and let $I_{N_k}$ be the corresponding sets as defined earlier. 
Assume that the compensation set is finite with indices $\{-N,...,N\}\backslash I_{N_k}$. 
If 
\begin{equation*}
\limsup_{k\rightarrow \infty} \frac{|I_{N_k}|}{N_k}>0,
\end{equation*}
then the PAPR problem is not solvable for the sequence of information sets $\{e^{2\pi i l\cdot}\}_{l\in I_{N_k}}$. 
\end{corollary}
\begin{proof}
The proof of the previous theorem gives that inequality~(\ref{discretevector}) cannot hold for a common constant. 
Since the discrete case gives the values of the continuous case on the regular $N$-point grid, 
it follows that there does not exist a universal constant such that 
\begin{equation}
\left\| \sum_{l=1}^N e^{-2\pi il\cdot}x_l+ \sum_{l=1}^N e^{-2\pi il \cdot}y_l\right\|_{L^\infty([0,1])}\leq C_{\Ex}\|x\|_{l^2(I_N)},
\end{equation}
where $\supp(x)\subset I_N$ and $\supp(y)\subset R_N$. 
\end{proof}

Note that the results of this section also give a bound on the size of subsets for which one does have solvability. 
Namely, from Theorem~\ref{equivdiscrete} and Theorem~\ref{discretenotsolvable}, we see that solvability with constant $C_{\Ex}$ implies that 
if $m$ is the length of the longest arithmetic progress in an information set $I_N$, 
then $\frac{\sqrt{m}}{\log m}\leq C_{\Ex}$. 
Thus, for a given index set $I_N$, one can determine its long arithmetic progression and obtain a lower bound on the 
extension norm. 
\end{subsection}

\begin{subsection}{Examples of Solvability and Projection Properties}
The next theorem gives a condition on a subset of $\N$ such that one has a bounded extension constant. 
The sacrifice, however, is that the information set has a density converging exponentially fast to $0$. 
\begin{theorem}\label{solvelambda}
Let $R_N=\{r_l\}_{l=1}^L$ be subset of $\{1,...,N\}$ satisfying $r_k\geq \lambda r_{k-1}$ for some $\lambda>1$, 
and where $N$ is chosen such that $N=\lambda r_L$.  
There exists a constant $C(\lambda)$ depending only on $\lambda$ such that for all $a\in l^2_N$ supported on $R_N$ 
\begin{equation*}
\|a\|_{l^2_N}\leq \frac{C(\lambda)}{\sqrt{N}}\|Fa\|_{l^1_N}. 
\end{equation*}
\end{theorem}
The theorem should be compared with Theorem~\ref{diadsubsetwalsh}, its counterpart for the Walsh setting. 
In the theorem just stated, $\lambda$ provides for both  a proportion of 
functions with frequencies higher than $r_L$, as well as the solvability in the 
first place. 
This latter aspect is due to the following result of Paley. 
\begin{theorem}\emph{(Section I.B.8 of \cite{Woj91})} \label{Paley}
Let $0<p<\infty$ and let $\{n_k\}_{k=1}^\infty$ be a subsequence of  $\N$ such that $\inf_{k\in \N} \frac{n_{k+1}}{n_k}=\lambda>1$. 
Then there exist constants $A(\lambda,p)$ and $B(\lambda,p)$ such that for all sequences with 
only finitely many non-zero terms 
\begin{eqnarray*}
A(\lambda,p) \|\sum_{i=k}^\infty a_ke^{ik\cdot}\|_{L^p([0,1])} &\leq&\| \sum_{k=1}^\infty a_ke^{ik\cdot} \|_{L^2([0,1])} \leq B(\lambda,p) \|\sum_{k=1}^\infty a_ke^{ik\cdot}\|_{L^p([0,1])}.\nonumber
\end{eqnarray*}
\end{theorem}

\begin{proof}{\bf of Theorem~\ref{solvelambda}}
Assume that $a\in l^2_N$ is supported on $R_N$, and define 
\begin{equation*}
f(t)=\sum_{l=1}^L a_{r_l}e^{ir_lt}.
\end{equation*}
We define a set of kernels that allow us to represent $f$ in terms of its samples: 
\begin{eqnarray*}
\mathcal{K}_{N,\lambda}\hspace{-.2cm}&=&\hspace{-.2cm}\bigg\{K(t)=\sum_{k=-r_L}^{r_L}e^{ikt}+ \sum_{k=-N}^{-r_L-1}d_ke^{ikt}+\sum_{k=r_L+1}^{N}d_ke^{ikt},\nonumber\\
&&\hspace{1.5cm}\textnormal{ where }d_k=d_{-k} \textnormal{ for }k=r_L+1,...,N\bigg\}.
\end{eqnarray*}

For any kernel $K\in \mathcal{K}_{\lambda,N}$, we have 
\begin{equation*}
f(t)=\frac{1}{N}\sum_{l=0}^Nf\left(\frac{2\pi l}{N}\right)K\left(t-\frac{2\pi l}{N}\right).
\end{equation*}
Then, using Theorem~\ref{Paley}, 
\begin{eqnarray}
\|a\|_{l^2_N}&=& \|f\|_{L^2([0,1])}\label{5656}\\
&\leq & B(\lambda,1)\|f\|_{L^1([0,1])}\nonumber\\
&=& B(\lambda,1) \left\|\frac{1}{N}\sum_{l=0}^Nf\left(\frac{2\pi l}{N}\right)K\left(\cdot-\frac{2\pi l}{N}\right)   \right\|_{L^1([0,1])}\nonumber\\
&\leq &  \frac{B(\lambda,1)}{N}\sum_{l=0}^N\left|f\left(\frac{2\pi l}{N}\right)\right| \left\|K\left(\cdot-\frac{2\pi l}{N}\right)  \right\|_{L^1([0,1])}\nonumber\\
&= & \frac{B(\lambda,1)}{N}\sum_{l=0}^N\left|f\left(\frac{2\pi l}{N}\right)\right| \|K \|_{L^1([0,1])}\label{1111}\\
&=& B(\lambda,1)\|K \|_{L^1([0,1])} \frac{1}{N}\sum_{l=0}^N|f\left(\frac{2\pi l}{N}\right)|,\nonumber
\end{eqnarray}
where in line~(\ref{1111}) we use the shift-invariance of the $\|\cdot\|_{L^1(\T)}$-norm for any $K$. 

Setting $d_k=(N-r_L-\frac{k}{2})$ for $k=r_L,..., N$ and $K(t)=\frac{1}{N-r_L} ( \sum_{k=0}^{N}D_k(t) -\sum_{k=0}^{r_L}D_k(t))$, 
we may define $f+g_{\lambda N}$ to be the convolution of $f+g$ with $K$. 
The Fourier expansion of $f+g_{ N}$ is supported on $\{-N,...,N\}$ 
and agrees with $a$ on $\{r_k\}_{k=1}^\infty$. 
Using $P_K$ to denote the projection given by convolution with $K$, 
\begin{equation}
\|f+g_{N}\|_{L^\infty(\T)}\leq \|P_K\| \|f+g\|_{L^\infty(\T)}\leq C_{Ex}\|a\|_{l^2(\Z)}\|P_K\| .\label{56}
\end{equation}
The norm $ \|P_K\| $ is the $\|\cdot\|_{L^1(\T)}$-norm of $K$. 
We will construct $K$ using two Fej\'er kernels. 
We recall that the Dirichlet kernel is defined by 
\begin{equation*}
D_n(t)=\sum_{k=-n}^n e^{ikt}, 
\end{equation*}
and the Fej\'er kernel by 
\begin{equation}
F_n(t)=\frac{1}{n}\sum_{k=0}^{n-1}D_n=\left(\frac{\sin \frac{nt}{2}}{\sin \frac{t}{2}} \right)^2.\label{fejpos1}
\end{equation}

Thus, for any $m>l$,  
\begin{eqnarray*}
\sum_{k=0}^m D_k-\sum_{k=0}^lD_k &=& (m-l)\sum_{k=0}^{2l-m}(e^{ikt}+e^{-ikt})\\
&&\hspace{.5cm} +\sum_{k=1}^{2(m-l)}\left(m-l-\frac{k}{2}\right)\left(e^{i(2l-m+k)t}+e^{-i(2l-m+k)t}\right). 
\end{eqnarray*}

Using the positivity given in equation~(\ref{fejpos1}), 
we obtain 
\begin{eqnarray}
\|K \|_{L^1([0,1])} &=& \int_0^1\left|\frac{1}{N-r_L} \left( \sum_{k=0}^{N}D_k(t) -\sum_{k=0}^{r_L}D_k(t)\right)\right|dt\nonumber\\
&\leq & \frac{1}{N-r_L}\int_0^1\sum_{k=0}^{\lambda r_L}D_k(t) +\sum_{k=0}^{r_L}D_k(t)dt\nonumber\\
&\leq & \frac{2 N}{N-r_L}\nonumber\\
&\leq & \frac{2\lambda }{\lambda-1}.\label{121234}
\end{eqnarray}
Returning to line~(\ref{56}), we have 
\begin{equation*}
\|f+g_{N}\|_{L^\infty(\T)}\leq \frac{2\lambda}{\lambda-1}C_{Ex}\|a\|_{l^2(\Z)},
\end{equation*}
where the Fourier expansion of $g_{N}$ is supported on $\{N,...,N\}\backslash I_N$.

Returning to equation~(\ref{5656}), we now have
\begin{eqnarray*}
\|a\|_{l^2_N}&\leq &  \frac{2\lambda B(\lambda,1)}{\lambda-1} \frac{1}{N}\sum_{l=0}^N\left|f\left(\frac{2\pi l}{N}\right)\right|\\
&=&  \frac{2 \lambda B(\lambda,1)}{\lambda-1} \frac{1}{\sqrt{N}}\sum_{l=0}^N\left| \frac{1}{\sqrt{N}}\sum_{k=1}^{L}a_{r_k}e^{2\pi ir_k \frac{l}{N}} \right|\\
&=&  \frac{2\lambda B(\lambda,1)}{\lambda-1} \frac{1}{\sqrt{N}}\|Fa\|_{l^1_{N}}.
\end{eqnarray*}
\end{proof}

\end{subsection}

\end{section}

\begin{section}{The Walsh or CDMA Case}
\begin{subsection}{PAPR and Density for Walsh or CDMA Systems}\label{walsh}
PAPR has been studied extensively in OFDM systems. 
In this section we show that the same type of behavior occurs as well in the down-link 
of Direct Sequence Code Division Multiple Access (DS-CDMA) systems. 
We assume again, without loss of generality, that the symbol period is normalized to length $1$. 
Assume that a base station communicates with $N=2^n$ users. 
DS-CDMA then uses $2^n$ orthogonal functions which take the values $1$ or $-1$ on 
intervals of length $2^{-n}$. 
These functions are the spreading sequences. 
We denote these sequences $\{w_k\}_{k=1}^{2^n}$, and will give their specific values shortly. 
The base station then transmits 
\begin{equation}
s(t)=\sum_{k=1}^{2^n} d_kw_k(t)\;\;\;\textnormal{ for }t\in[0,1]\label{cdma}
\end{equation}
to transmit the coefficient $d_k$ to user $k$. 

However, Proposition~\ref{sqN} applies here, and so  the function in equation~(\ref{cdma}) 
can achieve peaks of size $\sqrt{2^n}$. 
We note that in the up-link, each user only transmits one signal $w_k$, and  so there is not the accumulation of signals 
that leads to the high peaks that can occur in the down-link. 
Thus, in order to reduce the PAPR of DS-CDMA down-link signals one could reserve certain spreading sequences  to be used 
for compensation in analogy to tone reservation for OFDM. 

There are several ways to define the Walsh system, though the definitions only involve a different ordering. 
We present one definition now, and will comment on another in Section~\ref{matrixsection}. 
The various definitions made be found in the first several pages of~\cite{SWS90}. 
The following is the original ordering given by Walsh~\cite{Wal23}, 
and is the system used in the down-link for IS-95 standard and UMTS-IMT-2000.
\begin{definition}
The \emph{Rademacher functions}, denoted $r_0,r_1,...$,  are defined on $[0,1]$ by 
\begin{equation}
r_k(t)=\sign \sin (2\pi 2^k t),
\end{equation}
where we define $\sign\; 0 =-1$. 
The \emph{Walsh functions}, denoted $w_1,w_2,...$  are defined using the Rademacher functions by 
\begin{equation}
w_{1}(t)=1
\end{equation}
and
\begin{equation}
w_{2^k+m}(t)=r_k(t)\cdot w_m(t)
\end{equation}
for $k=0,1,2,...$ and $m=1,...,2^k$. 
\end{definition}
See \cite{Fin49} for the fundamental properties of the Walsh functions. 

We are able to obtain specific bounds on $C_{\Ex}$ in the Walsh case and, therefore, can state more precise results than in the discrete Fourier case. 
Our main result concerning the Walsh functions is the following. 

\begin{theorem}\label{mainwalsh}
Let $\delta\in (0,1)$ be a density and assume that $N=2^n$ $(n\in \N)$ satisfies $N\geq (\frac{2}{\delta})^{m+1}$ for some $m\in \N$.  
If the PAPR problem is solvable with constant $C_{\Ex}$ for a subset of indices $Y\subset \{1,...,N\}$ for $|Y|/N\geq \delta$, then 
\begin{equation}
C_{\Ex}\geq \frac{2^m-m^2}{1+m}.
\end{equation}
\end{theorem}

There are several necessary lemmas before we can prove the theorem. 
We begin with a definition. 

\begin{definition}
The \emph{correlation function} between $w_r$ and $I$ is 
\begin{equation}
C(w_r,I)=\Izo w_r(x)|\skI w_k(x)|^2dx.
\end{equation}
Further, for $w_r$ and $I$ we define the following set
\begin{equation}
\M(w_r,I)=\{ k\in I: \textnormal{ there exists } \kt \in I \textnormal{ so that } w_k(x)w_{\kt}(x)=w_r(x) \textnormal{ for all }  x\in [0,1]\}.
\end{equation}
\end{definition}
We could equivalently look at the set of all pairs $(k,\kt)$ such that $w_kw_{\kt}=w_r$, 
which would always include the permutation $(\kt,k)$. 
For every $k\in \MrI$, there is exactly one $\kt$ in $\MrI$ such that the pair $k$ and $\kt$ satisfy the requirement given for $\MrI$. 
To see this, suppose that $\kt_1$ and $\kt_2$ both satisfy 
\begin{equation}
w_kw_{\kt_1}=w_r=w_kw_{\kt_2}. \label{kkt}
\end{equation}
Then, since $w_kw_k=1$ for any $k$, multiplying (\ref{kkt}) by $w_k$ we have 
\begin{equation}
w_{\kt_1}=w_{\kt_2}.
\end{equation}

The reader unfamiliar with the Walsh functions may find it helpful to read the proof of Theorem~\ref{walsh-proj}. 
The properties used in the proof below are contained in the latter proof. 
In particular, the central idea of the proof of the main lemma, Lemma~\ref{mainlemmawalsh}, relies on the identity in Equality~(\ref{987}). 
This Equality gives a representation of the sum of Walsh functions as a product of factors $(1+r_k)$ for the appropriate Rademacher functions. 
This property and the fact that multiplying the set of Walsh functions (other than $w_1$) by a Walsh function gives a permutation of the 
Walsh functions lead to the idea of expressing a linear combination of Walsh functions as a product. 
It this is done correctly, one can obtain the $L^1$ and $L^2$ properties necessary for the theorem. 
\begin{lemma}\label{lemma4.4}
The following equality holds
\begin{equation}
C(w_r,I)=|\MrI|.
\end{equation}
\end{lemma}
\begin{proof}
\begin{eqnarray}
C(w_r,I)&=& \Izo w_r(x)\left( \sum_{k_1\in I} w_{k_1}(x)\right)\left( \sum_{k_2\in I} w_{k_2}(x)\right)dx\\
&=& \sum_{k_1\in I}\sum_{k_2\in I}\Izo w_r(x)w_{k_1}(x)w_{k_2}(x)dx.\label{rk1k2}
\end{eqnarray}
For $1\leq r,k_1,k_2\leq N$, $w_rw_{k_1}w_{k_2}=w_{\tilde{r}}$ for some $1\leq \tilde{r}\leq N$. 
Therefore the integral in (\ref{rk1k2}) is only nonzero when $w_rw_{k_1}w_{k_2}=w_1=1$, 
that is, when $w_r=w_{k_1}w_{k_2}$. 
This is the  set of all pairs of $k_1$ and $k_2$ such that this holds.  
In light of the comment following the definition of $\MrI$, the lemma is proved.
\end{proof}
An important monotonicity property follows from this lemma. 
\begin{corollary}
Assume  that $C(r,I)=0$. 
Then $C(r,\tilde{I})=0$ for all $\tilde{I}\subset I$. 
\end{corollary}
\begin{proof}
This follows from
\begin{equation}
0=C(w_r,I)=|\MrI|\geq |\M(w_r,\tilde{I})|=C(w_r,\tilde{I}).
\end{equation}
\end{proof}
\begin{lemma}\label{sumC}
If $N=2^n$ for a positive integer $n$, then 
\begin{equation}
\sum_{r=1}^NC(w_r,I)=|I|^2.
\end{equation}
\end{lemma}
\begin{proof}
Recall that for $x\in [0,\frac{1}{N})$, $w_r(x)=1$ for all $r$. 
And, since $N=2^n$, $\sum_{r=1}^Nw_r(x)=0$ for all $x\in [\frac{1}{N},1]$. 
See, for example, \cite{Fin49}.  
Therefore, 
\begin{eqnarray}
\sum_{r=1}^NC(w_r,I)&=& \Izo \sum_{r=1}^Nw_r(x)|\skI w_k(x)|^2dx\\
&=& N\int_0^{\frac{1}{N}}|\skI w_k(x)|^2dx\\
&=& N \int_0^{\frac{1}{N}}|I|^2dx\\
&=& |I|^2.
\end{eqnarray}
\end{proof}

\begin{definition}
If $f$ is a linear combination of Walsh functions we define $\supp(f)$ to be the set of indices corresponding 
to the functions in the sum. 
\end{definition}

\begin{lemma}\label{mainlemmawalsh}
Let $\delta$ be a density, that is $\delta\in (0,1)$. 
For any $m>1$,  if  $N=2^n$ $(n\in \N)$ satisfies $N\geq (\frac{2}{\delta})^{m+1}$, 
then for any subset $I\subset \{1,...,N\}$ satisfying $\frac{|I|}{N}\geq \delta$, 
there exists a function $f$ supported on $I$ such that 
\begin{equation}
\|f\|_{L^2([0,1])}\geq 2^m-m^2\; \textnormal{ and } \;\|f\|_{L^1([0,1])}\leq 1+m.
\end{equation}
\end{lemma}

\begin{proof}
First we just take $N$ to be large and $I$ to be a subset of $\{1,...,N\}$. 
We will construct a function with the norm properties given, and then we show that $m$ can be made large in dependence on $\delta$ and $N$. 

Suppose that $\MrI=\{l_1,...,l_{2k}\}$ . 
Then $w_{r}\cdot\{ w_{l_1},...,w_{l_{2k}}\}=\{w_{l_1},...,w_{l_{2k}}\}$. 
Also, suppose that $w_{l_j}w_{l_{j+1}}=w_r$, for $j=1,...,k$. 
Then $w_rw_{j}=w_{j+1}$ for each $j$. 
Therefore, we may select the subset $\{w_{l_1},w_{l_3},..., w_{l_{2k-1}}\}$, 
and obtain the properties $w_r\cdot\{w_{l_1},w_{l_3},..., w_{l_{2k-1}}\}=\{w_{l_2},w_{l_4},..., w_{l_{2k}}\}$,  
and $w_r\cdot\{w_{l_1},w_{l_3},..., w_{l_{2k-1}}\}\cap\{w_{l_2},w_{l_4},..., w_{l_{2k}}\}=\emptyset$. 
We use this splitting to select a set that does not allow for any factoring of $\wo$. 
We use the notation $I^{(1)}_A$ and $I^{(1)}_B$ to denote a splitting of $I^{(1)}$ in the way just described.  
Since $C(w^{(1)},I^{(1)})=|\M (\wo, I^{(1)})|$ and $\wo$ cannot be realized as the product of any two Walsh functions with indices in 
$I^{(1)}$, we have $C(w^{(1)},I^{(1)})=0$. 
Lastly, we also have $|I^{(1)}|=\frac{1}{2}|\MrI|$.

Let $\wo$ be the Walsh function with the highest correlation with $I$, that is 
\begin{equation}
C(\wo,I)=\max_{w=w_2,...,w_N}C(w,I). 
\end{equation}
Since $C(w,I)=|\M(w,I)|$, 
the maximizer $\wo$ is the Walsh function whose index corresponds to the subset of $I$ with the greatest number 
of splittings contained in $I$. 

Now set $I^{(2)}=\M(\wo, I^{(1)})$, and define $I^{(2)}_A$ and $I^{(2)}_B$ accordingly. 
We repeat this until ultimately $I^{(m)}=\{i_{m,1},i_{m,2}\}$ for two indices such that $w_{i_{m,1}}w_{i_{m,2}}=w^{(m)}$. 
We have then selected Walsh functions $\wo,w^{(2)}, ...,w^{(m)}$ and 
determined sets 
$I^{(1)}_{A},I^{(1)}_{B} ,...,I^{(m)}_A,I^{(m)}_B$, 
such that 
\begin{equation}
w^{(j)}I^{(j)}_{A}=I^{(j)}_{B},\;I^{(j)}_{A}=w^{(j)}I^{(j)}_{B},\;w^{(j)}I^{(j)}_{A}\cap I^{(j)}_{B} =\emptyset \;\textnormal{ and }I^{(j)}_{A}\cup I^{(j)}_{B}=I^{(j)} \label{123456}
\textnormal{ for }1,...,m.
\end{equation}

Define 
\begin{equation}
\mathcal{F}_0=\bigg\{f: f(x)=\sum_{k\in I} \alpha_k w_k(x),\;\alpha_k\in \{0,1\},\;\textnormal{ at least one } \alpha_k\neq0\bigg\}.
\end{equation}
and for $l=1,...,m$ define
\begin{equation}
\mathcal{F}_l=\bigg\{f: f(x)=\sum_{k\in I^{(l)}} \alpha_k w_k(x),\;\alpha_k\in \{0,1\},\;\textnormal{ at least one } \alpha_k\neq0\bigg\}.
\end{equation}

Now we build up our desired function. 
We have $w^{(m)}=w_{i_{m,1}}w_{i_{m,2}}$,  so that $ w^{(m)}w_{i_{m,2}}=w_{i_{m,1}}$. 
Then
\begin{equation}
(1+w_{i_{m,2}})w^{(m)}=w_{i_{m,1}}+w_{i_{m,2}}\in I^{(m)}\subset I^{(m-1)}_A.\label{12345}
\end{equation}
The two Walsh functions in the sum (\ref{12345}) are unique. 
Thus $(1+w_{i_{m,2}})w^{(m)}\in \F^{(m)}$, and we set $f^{(m)}=(1+w_{i_{m,2}})w^{(m)}$. 
Now we repeat this by looking at
\begin{equation}
(1+w^{(m-1)})f^{(m)}.
\end{equation}
By the properties given in lines (\ref{123456}) and (\ref{12345}), we have $\supp(w^{(m-1)}f^{(m)})\subset I^{(m-1)}_B$, 
while $\supp(f^{(m)})\subset I^{(m-1)}_A$. 
Thus $(1+w^{(m-1)})f^{(m)}$ has four unique terms and is contained in $\F^{(m-1)}$. 
We denote it $f^{(m-1)}$. 
Similarly, $(1+w^{(m-2)})f^{(m-2)}$ has eight unique terms, and is contained in $\F^{(m-2)}$. 
We continue this and ultimately arrive at $f^{(0)}\in \F^{(0)}$. 
$f^{(0)}$ is the sum of $2^{m}$ unique Walsh functions, 
and, in particular, 
\begin{equation}
1+f^{(0)}=\prod_{l=1}^{m}(1+w^{(l)})-\sum_{l=1}^{m}w^{(l)}. 
\end{equation}
This delivers the bounds
\begin{eqnarray*}
\Izo |1+f^{(0)}(x)|^2dx&=&   \Izo \left|\prod_{l=1}^{m}(1+w^{(l)}(x))-\sum_{l=1}^{m}w^{(l)} \right|^2dx\\
&\geq &   \Izo \left|\prod_{l=1}^{m}(1+w^{(l)}(x)) \right|^2dx -m^2\\
&=&\Izo \prod_{l=1}^{m}|1+w^{(l)}(x)|^2dx -m^2\\
&=& \Izo \prod_{l=1}^{m}2(1+w^{(l)}(x))dx-m^2\\
&=&2^{m}-m^2
\end{eqnarray*}
and 
\begin{eqnarray*}
\Izo |1+f^{(0)}(x)|dx&=&   \Izo \left|\prod_{l=1}^{m}(1+w^{(l)}(x))-\sum_{l=1}^{m}w^{(l)} \right|dx\\
&\leq &   \Izo \left|\prod_{l=1}^{m}(1+w^{(l)}(x)) \right|dx +m\\
&=&\Izo \prod_{l=1}^{m}(1+w^{(l)}(x))dx +m\\
&=&1+m.
\end{eqnarray*}

Now it remains to obtain a lower bound on $m$. 
We have 
\begin{equation}
|I^{(1)}|= \half |\M(\wo,I)|.
\end{equation}

Using Lemma~\ref{sumC}, 
\begin{eqnarray*}
|I^{(1)}|&=& \frac{1}{2}|\M(\wo,I)|\\
&=& \frac{1}{2}\max_{w_r\neq w_1}C(w_r,I)\\
&\geq & \half \frac{1}{N}\sum_{r=1}^NC(w_r,I)\\
&=& \half \frac{|I|^2}{N}\\
&\geq&\frac{\delta}{2}|I|. 
\end{eqnarray*}
Similarly, 
\begin{equation}
|I^{(k)}|\geq \frac{\delta}{2}|I^{(k-1)}|\;\;\textnormal{ when  }\; \;|I^{(k-1)}|\geq 4. 
\end{equation}
If the process goes until $m$, then  $|I^{(m)}|\geq (\frac{\delta}{2})^{m}\delta N$. 
If $N\geq (\frac{2}{\delta})^{m+1}$, then $ (\frac{\delta}{2})^{m}\delta N\geq 2$, and, thus, 
the set splitting can be performed $m$ times. 
\end{proof}

\begin{proof}{\bf of Theorem~\ref{mainwalsh}}
By Corollary~\ref{corequiv}, solvability implies that for all $f$ with support in $Y$ we have 
\begin{equation}
\|f\|_{L^2(\T)}\leq C_{\Ex}\|f\|_{L^1(\T)}.
\end{equation}
By Lemma~\ref{mainlemmawalsh}, we have $C_{\Ex}\geq \frac{2^m-m^2}{1+m}$.
\end{proof}

We return to the statement of Theorem~\ref{mainwalsh}. 
Suppose that a certain $\delta$ is fixed. 
Then as $N$ increases, the largest value for $m$ that still satisfies $N\geq (\frac{2}{\delta})^{m+1}$ increases. 
Thus the lower bound on $C_{\Ex}$ increases as $N$ increases. 
That is, given a density $\delta$ and a number of Walsh functions $N$, one knows a lower bound for the best 
possible extension norm. 
In particular, for a fixed $\delta$ there does not exist a uniform extension constant. 

As commented  before Lemma~\ref{lemma4.4}, the two fundamental properties here are the permutation of the Walsh functions 
when multiplied by another Walsh function and the representation of the sum as a product. 
Indeed, we only needed to find $m$ Walsh functions such that their products are all unique and supported in $I$. 
The $m$ functions also need not have indices in $I$. 
The product terms then disappear for the $L^1$-norm but not for the $L^2$-norm. 
The first steps of this approach also work for the Fourier case, since the properties just listed hold there as well. 
One can also define a correlation function and look for splittings of maximum cardinality. 
The approach encounters a difficulty, though, when trying to obtain an upper $L^1$ bound on the product that was easily bounded in the Walsh case. 

\end{subsection}
\begin{subsection}{Solvability of PAPR for Direct Sequence CDMA and Localized Behavior of the Walsh System}\label{RW}
We have several positive results for cases when the density of the information set converges to zero. 
Here, the Fourier and Walsh systems have a property in common. 
In particular, similar to Theorem~\ref{solvelambda} in the Fourier case, we show that if the information set is the dyadic integers, 
then any information bearing signal can be compensated for by a signal supported on the remaining indices. 
This is Theorem~\ref{subsetwalsh}. 
Therefore, the Fourier and Walsh systems behave similarly as far as density and solvability is concerned, 
and differ in terms of their projection properties. 
En route to the last results we require several definitions. 
We finish the section with Theorem~\ref{hadamard}, which gives a matrix embedding for Hadamard matrices.

\begin{theorem}(Khintchine's Inequality, I.B.8 in \cite{Woj91})\label{Khin}
There exist constants $A_p$ and $B_p$, $0 < p<\infty $ such that for all finite sequences of scalars $(a_i)_{i=1}^n$, 
\begin{equation*}
A_p \left\|\sum_{i=1}^n a_ir_i\right\|_{L^p([0,1])} \leq \|a\|_{l^2(\N)}=\left\| \sum_{i=1}^n a_ir_i\right\|_{L^2([0,1])} \leq B_p \left\|\sum_{i=1}^n a_ir_i\right\|_{L^p([0,1])}.\nonumber
\end{equation*}
\end{theorem}

The Walsh system has a very special property: projection operators mapping bounded functions onto the span of the first $2^k$ Walsh functions 
are uniformly bounded with norm $1$. 
Thus, when we combine this projection property with Khintchine's Inequality, we obtain a 
statement for finite sets. 
While the following material may be found in~\cite{Fin49},  
we include the proof of Theorem~\ref{walsh-proj} for the reader's convenience and to emphasize the unique properties of the 
Walsh system. 
Theorems~\ref{subsetwalsh} to~\ref{Cunder} are easy to understand once one has read through the proof below. 
Lines~(\ref{mmm})-(\ref{mmmm}) of the proof show that the dyadic Lebesque constants are $1$ for the Walsh system, 
which is a fundamentally different behavior from the $\log N$ behavior that occurs in the Fourier case. 
\begin{theorem}\label{walsh-proj}
Let $P_{2^n}$ denote the projection onto $\{w_1,....,w_{2^n}\}$. 
Then 
\begin{equation}
\|P_{2^n}\|_{\calC([0,1])\rightarrow L^\infty([0,1])}=1
\end{equation}
for all $n\in \N$,  
and if $f\in \calC([0,1])$, 
then for all $x\in [0,1]$
\begin{equation*}
\lim_{n\rightarrow\infty}(P_{2^n}f)(x)=f(x).
\end{equation*}
\end{theorem}
\begin{proof}
The projection of $f\in L^2([0,1])$ onto $w_n$ is 
\begin{equation*}
c_n=\int_0^1 w_n(x)f(x)dx.
\end{equation*}
We consider the projection onto $\{w_k\}_{k=1}^{n}$ at the point $x$ and denote this $s_n(x,f)$: 
\begin{equation*}
s_n(x,f)= \sum_{k=1}^{n} c_k w_k(x).
\end{equation*}
Equivalently, 
 \begin{equation*}
s_n(x,f)=\int_0^1 \sum_{k=1}^{n}f(t)  w_k(t) w_k(x)dt.
\end{equation*}
Then 
 \begin{equation}
s_{2^n}(x,f)=\int_0^1 \sum_{k=1}^{2^n}f(t)  w_k(t)  w_k(x)dt.\label{9999}
\end{equation}
We are interested in the sum
\begin{equation*}
\sum_{k=1}^{2^n} w_k(t)  w_k(x);
\end{equation*}
however, this is just the sum of all possible products of 
$\{r_0(x)r_0(t),...,r_n(x)r_n(t)\}$, 
and so 
\begin{equation}
\sum_{k=1}^{2^n} w_k(t)  w_k(x)=\prod_{k=1}^{n}(1+r_k(x)r_k(t)).\label{987}
\end{equation}
If $x$ and $t$ are in the same dyadic interval of length $2^{-n}$, 
then $r_k(x)r_k(t)=1$ for $k=1,...,n$. 
But, if   there exists $k$ less than or equal $n$ such that 
$x$ and $t$ are not in the same dyadic interval of length $2^{-k}$, 
then $r_k(x)r_k(t)=-1$. This is due to the fact that one term must equal $1$ and the other must equal $-1$. 
In this case the product~(\ref{987}) must equal zero. 
Defining 
\begin{equation*}
D_{2^n}(x,t)=\sum_{k=1}^{2^n} w_k(t)  w_k(x).
\end{equation*}
We then have 
\begin{equation}
D_{2^n}(x,t) =
\left\{ \begin{array}{ll}
\hspace{-2mm}2^n&\hspace{-2mm}x,t\textnormal{ in the same dyadic interval of length }2^{-n}\\
\hspace{-2mm}0&\hspace{-2mm}x,t\textnormal{ not in the same dyadic interval of length }2^{-n}.
\end{array}\right. \label{333}
\end{equation}
Then
\begin{equation*}
(P_{2^n}f)(x)=\int_0^1f(t)D_{2^n}(x,t)dt.
\end{equation*}
If $x\in I_m:=[\frac{m-1}{2^n},\frac{m}{2^n}]$, then using equation~(\ref{333}), 
\begin{eqnarray}
\|P_{2^n}\|_{\calC([0,1])\rightarrow \calC([0,1])}
& =&\sup_{f\in \calC([0,1]),\|f\|_{L^\infty([0,1])}=1} \int_0^1f(t)D_{2^n}(x,t)dt\label{mmm}\\
&=& \int_{I_m}D_{2^n}(x,t)dt\nonumber\\
&=& 1.\label{mmmm}
\end{eqnarray}
This proves the first claim. 

For $x\in[0,1]$ define $\alpha_n(x)$ and $\beta_n(x)$  by taking them to satisfy the following statement for an appropriate  integer $m$:
\begin{equation*}
\alpha_n(x)=m2^{-n}\leq x<(m+1)2^{-n}=\beta_n(x).
\end{equation*}
Now, returning to equation~(\ref{9999}) we have 
\begin{eqnarray*}
s_{2^n}(x,f)&=&\int_0^1f(t)D_{2^n}(x,t)dt\\
&=& 2^n \int_{\alpha_n(x)}^{\beta_n(x)}f(t)dt\\
&=& \frac{F(\beta_n(x))-F(\alpha_n(x))}{\beta_n(x)-\alpha_n(x)},
\end{eqnarray*}
where $F(x)$ is an integral of $f(x)$. 
Since $f$ is assumed to be continuous, we have 
\begin{equation*}
\lim_{n\rightarrow\infty}(P_{2^n}f)(x)=f(x).
\end{equation*}
\end{proof}

We include here the following theorem, which is proved in~\cite{BF10}, in order to contrast the projection behavior of the Fourier and 
Walsh bases. 
The additional redundancy of a factor of $\lambda$ frequencies in the compensation set is necessary to obtain the theorem below. 
For a given extension norm, as the size of the information-bearing set increases, not only does the compensation set have to increase 
proportionally, but the set of extra frequencies included beyond the highest frequency must also grow proportionally. 
We will return to this when we discuss the behavior of the Walsh system, where one can project sharply. 

\begin{theorem}[\cite{BF10}]\label{finprojection}
Suppose that $I_N$ is a subset of $\{-N,...,N\}$ and that for every $a\in l^2(\Z)$ supported on 
$I_N$ the PAPR reduction problem is solvable with an extension sequence supported on $\Z\backslash I_N$ 
and with extension bound $C_{Ex}$. 
Assume $\lambda>1$ and that $\lambda N$ is an integer. 
Then the PAPR reduction problem is also solvable with an 
extension sequence supported on $\{-\lambda N,...,\lambda N\}\backslash I_N$ with extension constant $\frac{2\lambda}{\lambda-1}C_{Ex}$. 
\end{theorem}
The next theorem addresses the same question for the Walsh system. 
The significant point here is that with Walsh functions one may work with only the dyadic set that the information-bearing 
coefficients are contained in. 
There is nothing gained or lost by using or not using any Walsh functions beyond this dyadic set. 
\begin{theorem}\label{subsetwalsh}
Suppose that $I_{N}$ is a subset of $\{1,2,3,....,N\}$ and 
that for every function $f$ of the form 
\begin{equation*}
f(t)=\sum_{k\in I_{2^n}} a_kw_k(t),
\end{equation*}
there exists a compensation function with  coefficient vector $a$  supported on $\N\backslash I_N$, such that the combined 
signal has $\|\cdot\|_{L^\infty(\T)}$-norm at most $C_{Ex}\|a\|_{l^2(\N)}$. 
Let $n$ be the smallest integer such that $N\leq 2^n$. 
Then there exists a compensation function with coefficients supported on $\{k\}_{k=1}^{2^n}\backslash I_N$, 
such that the  $\|\cdot\|_{L^\infty(\T)}$-norm of the combined signal is still at most $C_{Ex}\|a\|_{l^2(\N)}$.
\end{theorem}
\begin{proof}
By Theorem~\ref{walsh-proj}, we may simply project the original combined function onto the span of $\{w_k\}_{k=1}^{2^n}$ 
and maintain the same norm. 
\end{proof}
Now we may consider a special case of Theorem~\ref{subsetwalsh}, for which we know that the PAPR reduction problem is solvable, 
namely when $I_N$ is the set of dyadic integers. 
\begin{theorem}\label{diadsubsetwalsh}
Let $B_1$ be the constant given in Khintchine's Inequality (Theorem~\ref{Khin}). 
Then for any function of the form 
\begin{equation*}
\sum_{k=1}^{n} a_{2^k}w_{2^k}(t),
\end{equation*}
there exists a vector $b\in l^2_{2^n}$ supported on $\{k\}_{k=1}^{2^n}\backslash \{2^k\}_{k=1}^{n}$ and with norm 
$\|b\|_{l^2_{2^n}}\leq B_1 \|a\|_{l^2_{2^n}}$ such that 
\begin{equation*}
\left\| \sum_{k=1}^{n} a_{2^k}w_{2^k}+   \sum_{ \{k\}_{k=1}^{2^n}\backslash \{2^k\}_{k=1}^{n} } b_{k}w_{k}\right\|_{L^\infty([0,1])}\leq B_1  \|a\|_{l^2_{2^n}}.
\end{equation*}
\end{theorem}
\begin{proof}
We may take the subset $K$ in Theorem~\ref{equivalence} to be $\{2^k\}_{k=1}^{n}$ 
(Theorem~\ref{equivalence} of course holds when $L^p(\T)$ is replaced by $L^p([0,1])$.) 
Khintchine's Inequality (Inequality~\ref{Khin}) gives the norm equivalence, and thus there exists a sequence 
$b\in \N\backslash \{2^k\}_{k=1}^n$, such that 
\begin{equation}
\left\| \sum_{k=1}^n a_{2^k}w_{2^k}+   \sum_{ \N\backslash \{2^k\}_{k=1}^n } b_{k}w_{k}\right\|_{L^\infty([0,1])}\leq B_1  \|a\|_{l^2_{2^n}}.\label{together}
\end{equation}
Applying Theorem~\ref{subsetwalsh} to the function in~(\ref{together}) proves the theorem. 
\end{proof}

We define an Optimal subset size for the Walsh system. 
\begin{definition}[Optimal subset size-Walsh]
\begin{eqnarray*}
\E_N(C_\Ex,W)&=&\max\{ |I_N|;\; I_N\subset\{1,...,N\},\textnormal{  such that }\\
&&\hspace{1cm}\textnormal{ PAPR is solvable for } I_N \textnormal{with constant }C_\Ex\},
\end{eqnarray*}
where here  $W$ refers to the Walsh system. 
\end{definition}

The following result shows that for a given extension constant, the efficiency of the optimal subset 
does not increase as the dimension increases. 
This means we have strict monotone convergence on dyadic subsets of $\N$. 
The result is stronger than an asymptotic statement: it holds for all dyadic Walsh subsets. 
\begin{theorem}
For the Walsh system, the optimal subsets satisfy the following inequality
\begin{equation*}
2\E_{2^m}(C_\Ex,W)\geq \E_{2^{m+1}}(C_\Ex,W)
\end{equation*}
for all constants $C_{\Ex}$. 
\end{theorem}
\begin{proof}
We will call a subset of $\{0,...,2^{m}\}$ an optimal subset for the constant $C_{\Ex}$ if  
the PAPR reduction problem is solvable for the subset with constant $C_{\Ex}$ and there is no 
other subset of greater cardinality for which this holds. 
Let $I^*_{2^{m+1}}$ denote  an optimal subset of $\{0,...,2^{m+1}\}$ and 
$I^*_{2^{m}}$ an optimal subset of $\{0,...,2^{m}\}$. 
First define 
\begin{equation*}
I^{(1)}= \{0,1,...,2^m\}\cap I^*_{2^{m+1}}. 
\end{equation*} 
For $c\in l^2(I^{(1)}) $, let $f_{c,2^{m+1}}$ denote the extension function in $\span\{ w_0,...,w_{2^{m+1}}\}$ 
that satisfies 
\begin{equation*}
\|   f_{c,2^{m+1}}\|_{ L^\infty([0,1]) }\leq C_{\Ex}\|c\|_{l^2(I^{(1)})}. 
\end{equation*}
By Theorem~\ref{walsh-proj} we also have 
\begin{equation*}
\| P_{2^m}  f_{c,2^{m+1}}\|_{ L^\infty([0,1]) }\leq\|   f_{c,2^{m+1}}\|_{ L^\infty([0,1]) }\leq C_{\Ex}\|c\|_{l^2(I^{(1)})}. 
\end{equation*}
So, if we define $ f_{c,2^{m}}= P_{2^m}  f_{c,2^{m+1}}$, then $ f_{c,2^{m}}$ is also a solution for 
$c\in l^2(I^{(1)}) $ with constant $C_{\Ex}$. 
Now, since we assumed that $I^*_{2^{m}}$ is an optimal subset, we must have 
$|I^{(1)}|\leq |I^*_{2^{m}}|$. 

Now we define $I^{(2)}$ by 
\begin{equation*}
I^{(2)}=\{2^m,...,2^{m+1}\}\cap I^*_{2^{m+1}}\}.
\end{equation*}
We define $P^{(2)}$ and $Q^{(2)}$ by 
\begin{eqnarray*}
(P^{(2)}f)(t)&=& (P_{2^{m+1}}f-P_{2^m})(t)\\
&=& \sum_{k=2^m}^{2^{m+1}} c_k(f)w_k(t)\\
&=& w_{2^m} \sum_{k=2^m}^{2^{m+1}} c_k(f)w_{k-2^m}(t)\\
&=& w_{2^m}(t)(Q^{(2)}f)(t).
\end{eqnarray*}
Note that $Q^{(2)}$ maps into $\span\{w_0,...,w_{2^m}\}$. 
Since $|w_{2^m}(t)|=1$, we have 
\begin{equation*}
\|P^{(2)}f\|_{ L^\infty([0,1]) }=\|Q^{(2)}f\|_{ L^\infty([0,1]) }.
\end{equation*} 
We can now do the same calculation for the Lebesque constants for $Q^{(2)}$ as was done in Theorem~\ref{walsh-proj} 
and determine that $\|Q^{(2)}\|=1$. 
Thus, 
\begin{equation*}
\|P^{(2)}f\|_{ L^\infty([0,1]) }= \|Q^{(2)}f\|_{ L^\infty([0,1]) }\leq  \|f\|_{ L^\infty([0,1]) }.
\end{equation*}
Let $c\in l^2(I^{(2)})$, and suppose that 
 $f^{(2)}_{c,2^{m+1}}$ is the extension function in $\span\{ w_0,...,w_{2^{m+1}}\}$ 
that satisfies 
\begin{equation*}
\|   f^{(2)}_{c,2^{m+1}}\|_{ L^\infty([0,1]) }\leq C_{\Ex}\|c\|_{l^2(I^{(2)})}. 
\end{equation*}
Then 
\begin{equation*}
\| Q^{(2)} f^{(2)}_{c,2^{m+1}}\|_{ L^\infty([0,1]) }\leq \|   f^{(2)}_{c,2^{m+1}}\|_{ L^\infty([0,1]) }\leq C_{\Ex}\|c\|_{l^2(I^{(2)})}.
\end{equation*}
Setting 
\begin{equation*}
\underline{I}^{(2)}=\{j\in\{0,...,2^m\};j+2^m\in I^{(2)}\},
\end{equation*}
since $Q^{(2)}$ maps onto $\span\{ 0,...,2^m\}$, 
we have shown solvability with the same constant on $\underline{I}^{(2)}$, 
so that  $|\underline{I}^{(2)}|\leq |I^*_{2^m}|$. 
Finally, 
\begin{eqnarray*}
|I^*_{2^{m+1}}| &=& |I^{(1)}\cup I^{(2)}|\\
&=&  |I^{(1)}|+| I^{(2)}|\\
&=&  |I^{(1)}|+| \underline{I}^{(2)}|\\
&\leq & 2 |I^*_{2^m}|. 
\end{eqnarray*}
\end{proof}

The following proposition is a simple consequence of the projection property for the Walsh system. 
It shows that if the information coefficients are all supported on a subset of $\{1,...,2^m\}$, 
then there is nothing to be gained by including Walsh functions with index higher than $2^m$ for the compensation. 
Thus, compensation in the Walsh system is a local problem for dyadic sets of integers. 
This is in strong contrast to the Fourier system, where, as can be seen in Theorem~\ref{finprojection}, it is necessary to 
include tones at a certain redundancy factor beyond the largest element of the subset.

We introduce two new terms. 
For $c\in l^2(I^*_{2^m})$ and $r\geq m$ we define 
\begin{eqnarray*}
\Interp(c,r,I^*_{2^m})=\{ f\;\in\; \span\{w_0,...,w_{2^r}\};\ c_k(f)=c_k\;\forall\; k\in L^*_{2^m}\} 
\end{eqnarray*}
and
\begin{equation*}
\underline{C}_{\Ex}(r,I^*_{2^m} )= \sup_{\|c\|_{l^2(I^*_{2^m})}}\left(\inf_{ f\;\in\; \Interp(c,r,I^*_{2^m})}\|f\|_{ L^\infty([0,1]) }\right).
\end{equation*}
\begin{proposition}\label{Cunder}
For all $r\geq m$, 
\begin{equation*}
\underline{C}_{\Ex}(r,I^*_{2^m} )= \underline{C}_{\Ex}(m,I^*_{2^m} ).
\end{equation*}
\end{proposition}
\begin{proof}
Suppose that 
\begin{equation}
\underline{C}_{\Ex}(r,I^*_{2^m} )< \underline{C}_{\Ex}(m,I^*_{2^m} ).\label{5432}
\end{equation}
Then for any $\epsilon>0$, 
\begin{equation*}
\inf_{ f\;\in\; \Interp(c,r,I^*_{2^m})}\|f\|_{ L^\infty([0,1]) }<\left( \underline{C}_{\Ex}(r,I^*_{2^m} )+\epsilon\right)\|c\|_{l^2(I^*_{2^m})}, 
\end{equation*}
and 
\begin{eqnarray*}
\inf_{ f\;\in\; \Interp(c,r,I^*_{2^m})}\|P_{2^m}f\|_{ L^\infty([0,1]) }&\leq& \inf_{ f\;\in\; \Interp(c,r,I^*_{2^m})}\|f\|_{ L^\infty([0,1]) }\\
&<&\left( \underline{C}_{\Ex}(r,I^*_{2^m} )+\epsilon\right)\|c\|_{l^2(I^*_{2^m})}, 
\end{eqnarray*}
which contradicts the assumption~(\ref{5432}). 
\end{proof}
\end{subsection}

\begin{subsection}{Matrix Results Related to the Walsh System}\label{matrixsection}
We give a result for matrices that is analogous to the DFT matrix for the Walsh system. 
As we have seen in throughout this section and the last, the Walsh system differs from the Fourier system in that results for the Walsh 
system hold for the dyadic sets of integers and do not require the redundancy needed in the Fourier case. 
This is seen again by comparing Proposition~\ref{hadamard} here with Proposition~\ref{solvelambda} in the Fourier case. 
\begin{definition}
The Rademacher matrix $R_k:\C^{k}\rightarrow \C^{2^k}$ is defined by 
\begin{equation*}
(R_{k})_{i,j}=\frac{1}{\sqrt{2^k}}r_{j}\left(\frac{i}{2^k}\right).
\end{equation*}
The \emph{Hadamard matrices} are defined inductively 
as follows
\begin{equation*}
H_1=[1],\hspace{2cm}
H_{2}=\frac{1}{\sqrt{2}}\left[ \begin{array}{lr}
1&1\\
1&-1
\end{array}
\right]
\end{equation*}
and
\begin{equation*}
H_{2^{k+1}}=\frac{1}{\sqrt{2^{k+1}}}\left[ \begin{array}{lr}
H_{2^k}&H_{2^k}\\
H_{2^k}&-H_{2^k}
\end{array}
\right]
\end{equation*}
for $k=2,3,...$.
\end{definition}
The Hadamard matrices are orthogonal and correspond to the Walsh system in that 
\begin{equation*}
(H_{2^k})_{i,j}=\frac{1}{\sqrt{2^k}}w_{j}(\frac{i}{2^k}).
\end{equation*}
Moreover, the Walsh system can be obtained from the Haar system by multiplying the Haar basis elements, represented as finite vectors, 
by the appropriate size Hadamard matrix. 
See section 1.4 of \cite{SWS90}. 

\begin{proposition}\label{hadamard}
There exists a constant $B_1$ such that for all $k$ 
\begin{equation*}
\|a\|_{l^2_k} \leq \frac{B_1}{\sqrt{2^k}} \|R_k a\|_{l^1_{2^k}}.
\end{equation*}
And, in dimension $2^k$, for all vectors $a$ supported on the 
set $D=1,2,4,8,....,2^k$, 
\begin{equation*}
\|a\|_{l^2_{2^k}} \leq \frac{B_1}{\sqrt{2^k}} \|H_{2^k} a\|_{l^1_{2^k}}.
\end{equation*}
\end{proposition}

\begin{proof}
We use that $r_1,...,r_k$ are constant on the intervals $(\frac{j-1}{2^k},\frac{j}{2^k} ]$ 
for $j=1,...,2^k$. 
Using the Khintchine inequality, we have 
\begin{eqnarray*}
\|a\|_{l^2_k}&=&\left\|\sum_{i=1}^n a_ir_i\right\|_{L^2([0,1])}\\
& \leq& B_1 \left\|\sum_{i=1}^n a_ir_i\right\|_{L^1([0,1])}\\
&=&B_1 \int_0^1\left| \sum_{i=1}^k a_ir_i(t)\right|dt\\
&=&B_1 \frac{1}{2^k}\sum_{l=1}^{2^k} \left| \sum_{i=1}^k a_ir_i( \frac{l}{2^k})\right|\\
&=& B_1 \frac{1}{\sqrt{2^k}}\sum_{l=1}^{2^k} \left|\frac{1}{\sqrt{2^k}} \sum_{i=1}^k a_ir_i( \frac{l}{2^k})\right|\\
&=& \frac{B_1}{\sqrt{2^k}}\|R_k a\|_{l^1_{2^k}}. 
\end{eqnarray*}
For the second statement, we simply note that for $a$ supported on 
$1,2,4,...,2^k$, 
\begin{equation}
\left\|\sum_{i=1}^{2^k} a_iw_i\right\|_{L^2([0,1])}=\left\|\sum_{i=1}^k a_{2^i}r_i\right\|_{L^2([0,1])},
\end{equation} 
so that the same calculation proves the claim. 
\end{proof}
\end{subsection}
\end{section}

\begin{section}{Conclusion and Discussion}
We have provided a contribution towards understanding the relationship between the peak values of a signal and the 
proportion of orthonormal signals that can be used for 
information transmission when using tone reservation for PAPR reduction. 
Our results show that for the two most common wireless systems, OFDM and DS-CDMA, a strict amplitude constraint requires that the proportion of 
signals used to carry information must decrease as the total number of signals used increases when using tone reservation. 

One could naively ask if this is the case for all orthonormal systems. 
However, we gave examples for both the Fourier and the Walsh case of subsequences, such that the corresponding subspaces have the norm 
equivalence, and thus the PAPR reduction problem is solvable for the infinite subsequence. 
By simply rearranging the original basis by alternatingly taking one function from the special subsequence and one from its complement, 
one has solvability on a subset with density $1/2$. 
We have seen, though, that when restricted to finite sets, one no longer has solvability in this setting. 
Thus, the behavior depends on the properties of the finite set. 
This can be seen for Walsh functions in the matrix setting as well. 
Suppose that one alternatingly selected a Rademacher function and a Walsh function that is not a Rademacher function and represented 
them as columns in a matrix. 
The number of rows would then grow exponentially with respect to the number of columns. 
Thus, the norm equivalence would not occur on spaces of the same dimension or even proportional dimension. 

In the Walsh case, we have seen that the three orderings for the Walsh system all yield the same result. 
This is because the necessary properties, namely that products only permute the functions within the appropriate dyadic block, 
are common to all the orderings. 
We state the informal conjecture that, as far as the topics addressed here are concerned, 
any basis with a useful structure (and uniformly bounded) will behave similarly to the Fourier 
and Walsh bases.

\end{section}
\vspace{.5cm}
\noindent{\bf \large Acknowledgment}

\noindent We note that the PAPR problem for CDMA systems was posed by Bernd
Haberland and Andreas Pascht of Bell Labs in 2000, and the first author thanks them for discussions since 
that time. 
We thank Andreas Kortke of the Technische Universit\"at Berlin and Wilhelm Keusgen 
of the Heinrich Hertz Institut for discussions 
concerning power amplifiers, in particular as part of the  Smart Radio Frontend Project. 
We thank Gerhard Wunder of the Heinrich Hertz Institut for help with the PAPR literature.
The approach taken to solve the PAPR reduction problem for CDMA was motivated by 
work of Timothy Gowers on quasirandom groups~\cite{Gow08}, which he  
discussed at his Institute for Advanced Studies lecture in summer 2010. 
The first author thanks the IAS for its hospitality during 
the  Workshop on Pseudorandomness in Mathematics 2010. 
Lastly, we mention that parts of this paper were presented in the first author's  Bell Labs Lecture in late 2010.

\def\cprime{$'$} \def\cprime{$'$} \def\cprime{$'$} \def\cprime{$'$}


\end{document}